\documentclass[final, letterpaper,journal ,onecolumn, 10 pt]{IEEEtran}

\usepackage{setspace} 

\usepackage{amsthm}
\usepackage{amsfonts}

\newtheorem{remark}{Remark}

\ifCLASSINFOpdf
\else
\fi
\usepackage[cmex10]{amsmath}
\usepackage{array}
\usepackage{graphicx}

\newtheorem{theorem}{Theorem}[]

\theoremstyle{definition}
\newtheorem{defn}{Definition}[]

\theoremstyle{remark}

\usepackage{graphicx}
\usepackage{epstopdf}
\DeclareGraphicsExtensions{.eps,.pdf}

\begin{document}
%
%
%
%
%
\title{A general framework for weighted sum-rate and common-rate optimization}

        \author{Koosha~Pourtahmasi Roshandeh$^{1*}$, Masoud~Ardakani$^2$ and Chintha~Tellambura$^3$\\
        	$^{1,2,3}$Department of Electrical and Computer Engineering, University of Alberta, Edmonton, Canada\\       
        	$^\star$ Email: pourtahm@ualberta.ca\

\thanks{}
\thanks{}
\thanks{}}

\maketitle

\begin{abstract}

In this paper, we propose a framework for solving a class of optimization problems encountered in a range of power allocation problems in wireless relay networks. In particular, power allocation for weighted sum-rate and common-rate optimization problems fall in this framework. Subject to some conditions on the region of feasible powers, the optimal solutions are analytically found. The optimization problems are posed in a general form and their solutions are shown to have applications in a number of practical scenarios. Numerical results verify the optimality of the analytical approach.

\end{abstract}

\begin{IEEEkeywords}
Weighted  sum-rate, common-rate, power allocation, optimization.
\end{IEEEkeywords}

%
\IEEEpeerreviewmaketitle

\section{Introduction}
%
%
%
%
 \IEEEPARstart{I}{n} multi-user  communication systems, weighted sum-rate and common-rate optimization problems have been extensively investigated. However, most existing results on weighted sum-rate or common-rate optimization are applicable to their specific setup. Usually, the common approach to the solution is through converting the problem to a convex one and then using standard convex optimization methods. In non-convex cases, approximate methods or heuristic algorithms are widely  used to find the near optimal solutions. Even in convex cases, the solutions are typically not in a closed-form.
 
 For instance, for single antenna transceivers, iterative algorithms for finding the optimal transmission power allocations to maximize weighted sum-rate under different quality of service (QoS) constraints in a wireless multi-link system have been presented in  \cite{1,2}. In \cite{3,4} and \cite{5}, authors have proposed approximate solutions by converting the non-convex weighted sum-rate and sum-rate maximization problems into the equivalent convex problems for multi-antenna and single-antenna sources, respectively. For a specific system model consists of  two single antenna transceivers and multiple relays, \cite{6} showed maximizing the sum-rate of two transceivers is equivalent to solving the common-rate optimization problem and proposed the optimal solution under the total power constraint. While \cite{6} offers analytical results, the solution is specific to its system model. 

%
%

In this letter, a general class of optimization problems which includes weighted sum-rate and common-rate optimization problems has been studied. We show that in a wide range of cases, the optimal solution can be found analytically, even when the feasible region of parameters is non-convex.  

In particular, for solving the weighted sum-rate and common-rate optimization problems in a multi-user network which consists of $N$ users,  one has to work with optimization problems with the general form of  $\max\limits_{}~    \sum_{i=1}^N  a_i\log_2(1+\gamma_{i})$ and $\max~ \min\limits_{i} (\gamma_{i})$, respectively, where $\gamma_i$ represents the SNR of user $i$ and $a_i$ is a constant weight. Also, some constraints on the region of feasible SNRs may exist, e.g.  due to power constraints. For this general setup, We propose the optimal solution to weighted sum-rate and common-rate maximization when the feasible region of SNRs satisfies some conditions (but need not be convex). 

Next, to show the practicality and usefulness of the considered framework and constraints, we first show that the results of \cite{6} can be restated under our framework. Then, we solve another practical example. More specifically, with the help of the proposed framework, the optimal power allocation for maximizing the weighted sum-rate and common-rate of a two-way relay network in which two users with single antenna communicate through a massive-antennas relay under sum-power constraint has been presented. Finally, simulation results have been represented to justify our framework and validate the theoretical analysis. Please note that the solution of the specific two-way relay network, considered here as an example, is a new result which does not appear in the literature. \\ 
\textit{\textbf{Notations}}: In this letter,   ${\bf h}^T$ and ${\bf h}^\dagger$  represent the transpose matrices and the Hermitian  of matrix ${\bf h}$, respectively.

\section{Framework}

Consider $N$ communication links. For example $N$ users trying to communicates with a base station or an ad-hoc network with $N$ users. The specific communication setup of these $N$ links is not our concern at this point. Let us denote the $\text{SNR}$ of $i$-th communication link at the destination by  $\displaystyle{\gamma_{i}}$. Assuming additive white Gaussian noise (AWGN), the weighted sum-rate optimization problem is defined as 

 \begin{equation}
 	\begin{cases}
 		\max\limits_{}~~~~    \sum_{i=1}^N  a_i\log_2(1+\gamma_{i})\\
 		\text {s.t.}~~~~~~~  \bar{\mathbf{\gamma}} \in \mathbf{\Theta} ,\\
 		
 	\end{cases}
 \end{equation}

 Where  $\bar{\mathbf{\gamma}}=(\gamma_1,\gamma_2,...,\gamma_n) \in  {\rm I\!R}^{+^n}$  is the vector of communication links $\text{SNR}$s  and   $\mathbf{\Theta}$ is the feasible region of achievable $\text{SNR}$s.
 
 The weighted sum-rate optimization problem can be further simplified as 

\begin{equation}
\begin{cases}
\max\limits_{}~~~~    \prod_{i=1}^N (1+\gamma_{i})^{a_i}\\
\text {s.t.}~~~~~~~  \bar{\mathbf{\gamma}} \in \mathbf{\Theta}.\\

\end{cases}
\end{equation}
Similarly, common-rate optimization problem in this general setup, can be expressed as

\begin{equation}
\begin{cases}
\max~~~~    \min\limits_{i} (\gamma_{i})\\
\text {s.t.}~~~~~~~  \bar{\mathbf{\gamma}} \in \mathbf{\Theta}.\\

\end{cases}
\end{equation}

To achieve the optimal solutions of mentioned problems, we first define some useful notations.



\begin{defn}
	
	For a given positive constant K and set $B=\{b_1,b_2...,b_n \} \subset {\rm I\!R}^+$ we define  $\mathbf{\Omega}_{B}(K)$ as the set of all points $\mathbf{X}=(X_1,X_2,...,X_n) \in {\rm I\!R}^{+^n}$ such that

	\[
	\sum_{i=1}^{n} b_i X_i \leq K.  
	\] 
\end{defn}  

\begin{defn}
	Let $A=\{a_1,a_2...,a_n \}$ and $B=\{b_1,b_2...,b_n \}$ be non-zero finite subsets of ${\rm I\!R}^+$ and K be a given positive constant. We define {\large $ \mathbf{\vartheta}_{A,B} ^K=({\vartheta}_{a_1,b_1} ^K,{\vartheta}_{a_2,b_2} ^K,...,{\vartheta}_{a_n,b_n} ^K) \in {\rm I\!R}^n$} where
\[	
	 {\large  \mathbf{\vartheta}_{a_i,b_i} ^K}= \displaystyle{\frac{a_i K}{b_i \sum_{i=1}^{n} a_i},~~ \forall i \in \{1,2,...,n \}}.	 
\] 
\end{defn}

One can easily show that for any arbitrary  $A \subset {\rm I\!R^+}$,  $ \mathbf{\vartheta}_{A,B} ^K \in \mathbf{\Omega}_{B}(K) $. We are now ready to propose the following theorems.


\begin{theorem}\label{Theo:sum rate}
Consider $\mathbf{X}=(X_1,X_2,...,X_n)  \in {\rm I\!R}^{+^n}$, $\mathbf{\Theta} \subset \mathbf{\Omega}_{B}(K)$ and $\mathbf{\vartheta}_{A,B} ^K \in \mathbf{\Theta}$. Then, $\displaystyle{\mathbf{X}=\mathbf{\vartheta}_{A,B} ^K}$ is the optimal solution of the following optimization problem
 \begin{equation}
\begin{cases}
\max~~~~    \prod_{i=1}^n X_i^{a_i}\\
\text{s.t.}~~~~~~~  \mathbf{X} \in \mathbf{\Theta}\\

\end{cases}
\end{equation}
 
\end{theorem}


\begin{proof}
Using definition (1) and geometric-mean arithmetic-mean inequality, we have:
	\begin{eqnarray}
	\frac{K}{\sum_{i=1}^{n} a_i} &\geq& \frac{\sum_{i=1}^{n} a_i \frac{b_i X_i}{a_i}}{a} \\&\geq&  \nonumber \sqrt[ a]{\prod_{i=1}^n (\frac{b_i}{a_i})^{a_i}X_i^{a_i}} 
	\end{eqnarray} 
where $a=\sum_{i=1}^{n} a_i$.\\ Geometric-mean achieves its upper bound (arithmatic-mean) if 
\[
	\displaystyle{\frac{b_i X_i}{a_i}=\frac{b_j X_j}{a_j}} ~~~ \forall i,j \in \{1,2,...,n\}
\] 
Since we want to maximize the geometric-mean the arithmatic mean should achieve its upper bound ($K$) simultaneously. Therefore,
\[ 
	\displaystyle{\sum_{i=1}^{n} b_i X_i=\sum_{i=1}^{n} a_i \frac{b_i X_i}{a_i}= \sum_{i=1}^{n} a_i t=K 
		\Rightarrow t=\frac{K}{\sum_{i=1}^{n} a_i}}, 
\]
which completes the proof.	
\end{proof}








The following Theorem proposes the optimal solution for a class of optimization problems which includes common-rate optimization problems.
\begin{theorem}
	\label{Theo:common rate}
The optimal solution of optimization problem 

\[
\begin{cases}
\max~~~~ \min \limits_{i} (X_i)\\
\text{s.t.}~~~~~~~  \mathbf{X} \in \mathbf{\Theta}  \\

\end{cases}
\]	
is  $\mathbf{Y} = (\frac{K}{\sum_{i=1}^{n} b_i },...,\frac{K}{\sum_{i=1}^{n} b_i } )$ if and only if $~\mathbf{Y} \in \mathbf{\Theta}$.

\end{theorem}


\begin{proof}	
Simply we have:
\begin{eqnarray}
\sum_{i=1}^{n} b_i \min{ \{X_i \}} &\leq& \sum_{i=1}^{n} b_i X_i \leq K \\
 \min{ \{X_i \}}&\leq&\frac{K}{\sum_{i=1}^{n} b_i } \nonumber
\end{eqnarray} 

To achieve the optimal solution we let $\displaystyle{\min{ \{X_i \}} = \frac{K}{\sum_{i=1}^{n} b_i }}$. \\Now if $\displaystyle{ \exists j \in  \{1,2,...,n\} ~~ X_j > \frac{K}{\sum_{i=1}^{n} b_i }}$
then
\begin{eqnarray}
\sum\limits_{i=1}^{n} b_i X_i  &\geq& \sum\limits_{i=1, i\neq j}^{n} b_i \min{ \{X_i \}} + b_j X_j \nonumber\\  
&>&  \frac{K(\sum\limits_{i=1, i\neq j}^{n} b_i)}{\sum\limits_{i=1 }^{n} b_i}+\frac{Kb_j}{\sum\limits_{i=1 }^{n} b_i} = K \nonumber
\end{eqnarray} 

which is a contradiction and completes the proof.
\end{proof}

Following Remarks reobtain the results in \cite{6} using the proposed Theorems.

\begin{remark}
	\label{remark 1}
For the rate region $\mathbf{\Theta}$ obtained in \cite{6},  $K=2+2 \gamma_{max}$, $b_i=1 ~ \forall i \in \{1,2\}$,	using Theorem \ref{Theo:sum rate}, the optimal solution for maximizing the sum-rate ($a_i=1 $ and $X_i=1+SNR_i~ \forall i \in \{1,2\}$ ) will be $SNR_i=\gamma_{max}~ \forall i \in \{1,2\}$ which is equal to the optimal solution obtained in \cite{6}.
\end{remark}

\begin{remark}
	\label{remark 2}
	Using Theorem \ref{Theo:common rate}, the optimal solution for maximizing the common-rate ( $X_i=SNR_i~ \forall i \in \{1,2\}$ ) for the rate region $\mathbf{\Theta}$ proposed in \cite{6}, can be obtained as $SNR_i=\gamma_{max}~ \forall i \in \{1,2\}$ which is equal to the optimal solution obtained in \cite{6}.
\end{remark}

\begin{remark}[]
From Remarks \ref{remark 1} and \ref{remark 2}, the  sum-rate and common-rate maximization problems defined for the obtained rate region in \cite{6} are equivalent which is one of the main results of \cite{6}.   
\end{remark}




\section{Applications}

\begin{figure}[t]
  \centering
  \vspace*{-2 cm}
    \includegraphics[width=9.5cm,height=15cm,keepaspectratio]{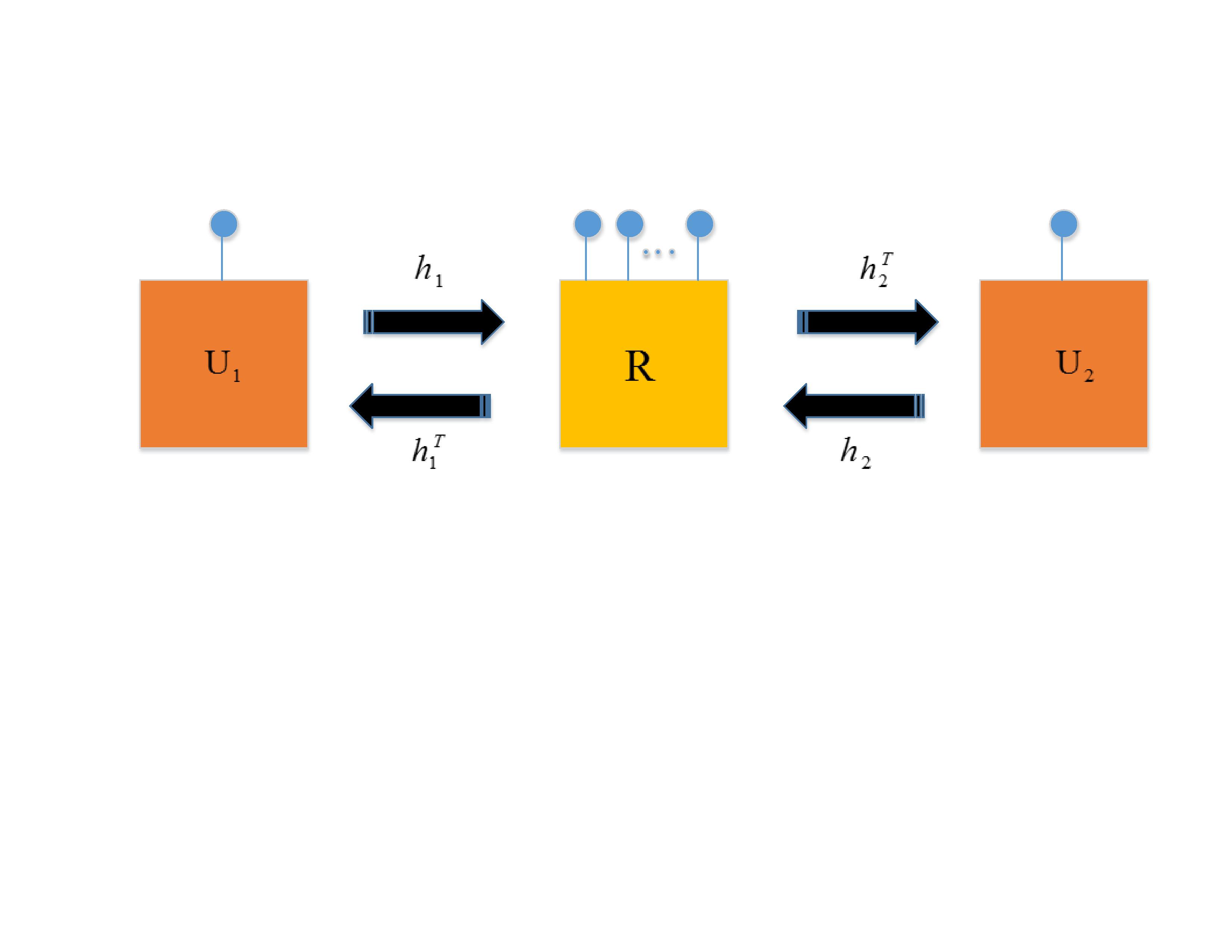}
    \vspace*{-4.2 cm}
    \caption{The system model of a two-way relay network using multiple-antenna relay. }
   
\end{figure}

In this section, to show the usefulness of the framework considered in the previous section, we propose a practical example. The optimization results that we present, to the best of our knowledge are novel and have not appeared in the literature. \\
We consider a two-way relay network including one relay which is equipped with $N_r$ antennas and two users $U_1$ and $U_2 $ with a single antenna. $\mathbf{h}_1$ and $\mathbf{h}_2 $ are the channel coefficients between $U_1$ and $U_2 $, and relay $R$. These coefficients are assumed to be reciprocal and independent. The transmit powers for users $U_a$, $U_b$ and relay $R$ are denoted by $P_1$, $P_2$ and $P_r$, respectively. Moreover, we assume additive white Gaussian noise (AWGN) with mean zero and variance $\sigma^2$  for each hop. The system model has been depicted in Fig. 1.
 The procedure of communication between two users is as follow:
 
 First, both users send their messages to the relay in the first time slot. In the second time slot, the relay transmits the combined messages to user $U_1$ (with the corresponding transmit weight vector for maximum ratio combining \cite{7} to $U_1$). In the third time slot, relay uses the corresponding transmit weight vector for maximum ratio combining to $U_2$ and sends the combined messages to $U_2$.
 

The received signal at $R$ is given by :
\begin{eqnarray}
y_r=\sqrt{P_1} \mathbf{h}_1 x_1 +\sqrt{P_2} \mathbf{h}_2 x_2 + n_r,  
\end{eqnarray} 
where unit energy transmit signals are denoted by $x_1$ and $x_2$ and $n_r$ is the noise at the relay.\
The transmit weight vectors for maximum ratio combining for the second and third time slots are as 
$ \mathbf{w}_1=(\frac{\mathbf{h}_2^\dagger}{\|\mathbf{h}_2\|})(\frac{\mathbf{h}_1^T}{\|\mathbf{h}_1\|})^\dagger$, $ \mathbf{w}_2=(\frac{\mathbf{h}_1^\dagger}{\|\mathbf{h}_1\|})(\frac{\mathbf{h}_2^T}{\|\mathbf{h}_2\|})^\dagger$ and the relay broadcasts the message with the gain:
\[
\displaystyle{k=\sqrt{\frac{P_r}{P_1\|\mathbf{h}_1\|^2+P_2\|\mathbf{h}_2\|^2+\sigma^2}}}
\]
After self-interference cancellation by each user, the received signal at $U_1$ and $U_2$ can be expressed as: 
\begin{eqnarray}\label{eqsignal}
\hat{y}_1 &=& k \sqrt{P_2} \|\mathbf{h}_1\|\|\mathbf{h}_2\|x_2 + k\|\mathbf{h}_1\|\hat{n}_{1r}+n_1 \nonumber\\
\hat{y}_2 &=& k \sqrt{P_1} \|\mathbf{h}_2\|\|\mathbf{h}_1\|x_1 + k\|\mathbf{h}_2\|\hat{n}_{2r}+n_2 ,
\end{eqnarray}
where $\hat{n}_{1r}= \frac{\mathbf{h}_2^\dagger \mathbf{n}_r}{\|\mathbf{h}_2\|}$, $\hat{n}_{2r}= \frac{\mathbf{h}_1^\dagger \mathbf{n}_r}{\|\mathbf{h}_1\|}  $ and $n_1$, $n_2$ are the AWGN noises at users $U_1$ and $U_2$, respectively.
Using \eqref{eqsignal}, one can easily show that the SNRs at the receiver of each user are given by: 
\begin{eqnarray}
\gamma_{2}&=&\frac{P_1 P_r}{\sigma^2}\Bigg[\frac{\|\mathbf{h}_1\|^2\|\mathbf{h}_2\|^2}{(P_2+P_r)\|\mathbf{h}_2\|^2+P_1\|\mathbf{h}_1\|^2+\sigma^2} \Bigg] 
\nonumber\\
\gamma_{1}&=&\frac{P_2 P_r}{\sigma^2}\Bigg[\frac{\|\mathbf{h}_1\|^2\|\mathbf{h}_2\|^2}{(P_1+P_r)\|\mathbf{h}_1\|^2+P_2\|\mathbf{h}_2\|^2+\sigma^2} \Bigg] 
.
\end{eqnarray}


\subsection{Common-rate}

We aim to maximize the common-rate subject to the total power constraint. Therefore, the optimization problem can be written as

\begin{equation}
\begin{cases}
\max\limits_{P_1,P_2,P_r}~~~~     \min (\gamma_1,\gamma_2)\\
\text {s.t.}~~~~~~~  P_1 + P_2 + P_r \leq P_t.\\	
\end{cases}
\end{equation}


In the next proposition, we reformulate the optimization problem (10).

\newtheorem{prop}{Proposition}
\begin{prop}
The constraint $\displaystyle{P_1 + P_2 + P_r \leq P_t}$ is equivalent to 
\[
		\displaystyle{\gamma_{1}+\gamma_{2}\leq\frac{\gamma_{1r}\gamma_{2r}}{\left(\sqrt{\gamma_{1r}+1}+\sqrt{\gamma_{2r}+1}\right)^2}},  
\]
where  $\displaystyle{\gamma_{1r}=\frac{P_t}{\sigma^2}\left| \left|  \mathbf{h}_1\right|\right|^2} $ , $\displaystyle{\gamma_{2r}=\frac{P_t}{\sigma^2}\left| \left|  \mathbf{h}_2 \right|\right|^2} $.
\end{prop}

\begin{proof}
We can define  $ P_1=\alpha \beta P_t$, $ P_2=(1-\alpha)\beta P_t$ and $ P_r=(1-\beta)P_t$ in which ( $0 \leq \alpha , \beta \leq 1  $). By substituting these definitions into the objective function $f(\alpha,\beta)=\gamma_{1}+\gamma_{2}$  and forming the equations $\displaystyle{\frac{\partial f}{\partial \beta}=0}$ and $\displaystyle{\frac{\partial f}{\partial \alpha}=0}$, 
the global maximum of function $f(\alpha,\beta)$ can be obtained as $\displaystyle{\frac{\gamma_{1r}\gamma_{2r}}{\left(\sqrt{\gamma_{1r}+1}+\sqrt{\gamma_{2r}+1}\right)^2}}.
$ 
\end{proof}
Hence, the optimization problem (10) turns to:
\begin{equation}
\begin{cases}
\max\limits_{(\gamma_{1},\gamma_{2}) \in \mathbf{\Theta}}~~~~      \min (\gamma_1,\gamma_2)\\
~~~\text {s.t.}~~~~~~~  \mathbf{\Theta} \subset \mathbf{\Omega}_{B}( \frac{\gamma_{1r}\gamma_{2r}}{\left(\sqrt{\gamma_{1r}+1}+\sqrt{\gamma_{2r}+1}\right)^2}) ,\\	
\end{cases}
\end{equation}
which is a special case of the optimization problem (3).


We present the optimal solution for problem (11) using Theorem (2) and show that the solution is in the subset $\mathbf{\Theta}$ and therefore, is feasible. 

Using Theorem (2) the optimal solution is 
\begin{eqnarray}
\gamma_{1}=\frac{\gamma_{1r}\gamma_{2r}}{2\left(\sqrt{\gamma_{1r}+1}+\sqrt{\gamma_{2r}+1}\right)^2}  \nonumber\\
\gamma_{2}=\frac{\gamma_{1r}\gamma_{2r}}{2\left(\sqrt{\gamma_{1r}+1}+\sqrt{\gamma_{2r}+1}\right)^2} \nonumber,
\end{eqnarray}

which results in 
\begin{align}
	\beta^{opt}&=0.5 \nonumber\\
	\alpha^{opt}&=\frac{-\gamma_{2r}-1+\sqrt{(\gamma_{2r}+1)(\gamma_{1r}+1)}}{\gamma_{1r}-\gamma_{2r}}, \nonumber
\end{align}
where both of them satisfy $0 \leq \alpha , \beta \leq 1  $. Substituting $\beta^{opt}$ and $\alpha^{opt}$ into $P_1, P_2, P_r$, we have:
\begin{eqnarray}
&P_1&=\frac{P_t\left(-\gamma_{2r}-1+\sqrt{(\gamma_{2r}+1)(\gamma_{1r}+1)}\right)}{2(\gamma_{1r}-\gamma_{2r})} \nonumber \\ &P_2&=\frac{P_t\left(\gamma_{1r}+1-\sqrt{(\gamma_{2r}+1)(\gamma_{1r}+1)}\right)}{2(\gamma_{1r}-\gamma_{2r})} \\ &P_r&=\frac{P_t}{2}  \nonumber
\end{eqnarray}
where 
$\displaystyle{\gamma_{1r}=\frac{P_t}{\sigma^2}\left| \left|  \mathbf{h}_1\right|\right|^2} $ , $\displaystyle{\gamma_{2r}=\frac{P_t}{\sigma^2}\left| \left|  \mathbf{h}_2 \right|\right|^2} $.

\subsection{Weighted sum-rate}






In this part, we aim to maximize the weighted sum-rate of the considered system model subject to the total power constraint. The weighted sum-rate of a two-users system can be expressed as  
\setlength{\arraycolsep}{0.0em}
\begin{eqnarray}
R&{}={}&\frac{a_1}{2}\log_2(1+\gamma_1)+\frac{a_2}{2}\log_2(1+\gamma_2)\nonumber\\
&&{=}\:\frac{1}{2} \log_2[(1+\gamma_1)^{a_1} (1+\gamma_2)^{a_2}
\end{eqnarray}
\setlength{\arraycolsep}{5pt}
Clearly, one can equivalently maximize the term $(1+\gamma_1)^{a_1} (1+\gamma_2)^{a_2}$. Hence, the optimization problem can be considered as
\begin{equation}
\begin{cases}
\max\limits_{P_1,P_2,P_r}~~~~      (1+\gamma_1)^{a_1} (1+\gamma_2)^{a_2}\\
~~~\text {s.t.}~~~~~~~  P_1 + P_2 + P_r \leq P_t,\\	
\end{cases}
\end{equation}
Using Proposition 1, optimization problem (18) turns to
\begin{equation}
\begin{cases}
\max\limits_{(1+\gamma_{1},1+\gamma_{2}) \in \mathbf{\Theta}}~~~~      (1+\gamma_1)^{a_1} (1+\gamma_2)^{a_2} \\
~~~~~~\text {s.t.}~~~~~~~~~~~~  \mathbf{\Theta} \subset \mathbf{\Omega}_{B}(2+ \frac{\gamma_{1r}\gamma_{2r}}{\left(\sqrt{\gamma_{1r}+1}+\sqrt{\gamma_{2r}+1}\right)^2}) ,\\	
\end{cases}
\end{equation}
which is a special case of the optimization problem (2).\\
From Theorem 1, the optimal solution of optimization problem (19) occurs when
\begin{eqnarray}
&\gamma_1&=\frac{a_1-a_2}{a_1+a_2}
+ \frac{a_1}{a_1+a_2}(\frac{\gamma_{1r}\gamma_{2r}}{\left(\sqrt{\gamma_{1r}+1}+\sqrt{\gamma_{2r}+1}\right)^2})
\nonumber\\
&\gamma_2&=\frac{a_2-a_1}{a_1+a_2}
+ \frac{a_2}{a_1+a_2}(\frac{\gamma_{1r}\gamma_{2r}}{\left(\sqrt{\gamma_{1r}+1}+\sqrt{\gamma_{2r}+1}\right)^2})\nonumber
\end{eqnarray}
For instance, assuming $a_1=2$, $a_2=1$, $P_t=0~dB$, $N_r=100$, $\gamma_{1r}=24 $,  and $\gamma_{2r}=96 $, the optimal solution will be $\gamma_1=7.3$ and $\gamma_2=3.15$ which translates to $P_1= (0.1996)P_t=0.1996,P_2=(0.2362)P_t=0.2362 $ and $P_r=(0.5642)P_t=0.5642$.

\begin{figure}[t!]
	\centering
	\includegraphics[width=21pc ]{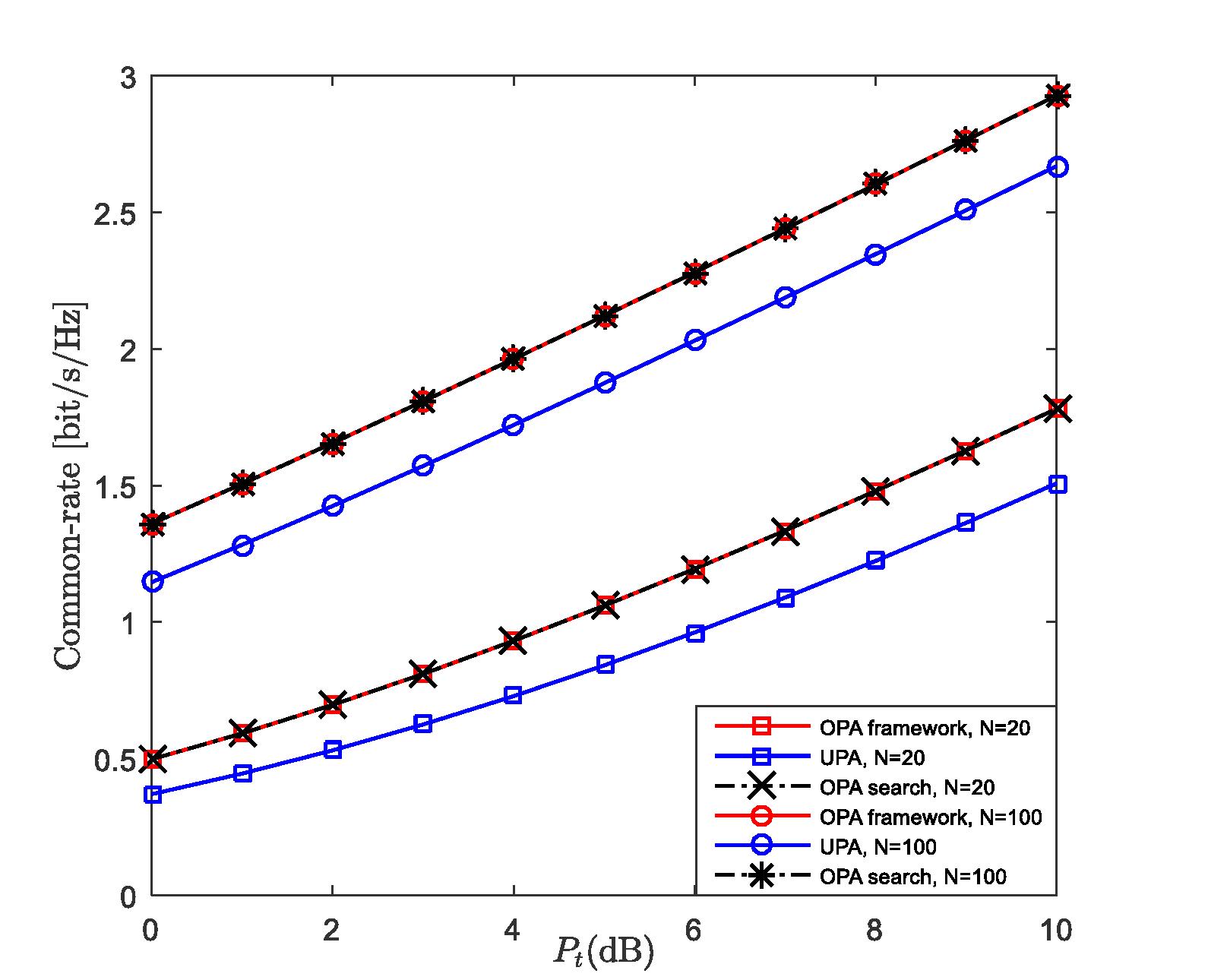}
	\caption{Achievable common-rate  $a_1=2, a_2=1, \sigma_1^2=0.25, \sigma_2^2=1$ and $\sigma^2=1 $.}	
\end{figure}

\section{Numerical and Simulation results}

In this section, simulation results have been presented to verify the optimality of the solutions proposed by Theorems 1 and 2. 


In Figs.~1 and 2, we have plotted the achievable common-rate and weighted sum-rate of the considered system model for the presented solutions in Section 3. To verify the optimality of these results, we also present the solution which has been obtained through searching (with step 0.001) the feasible SNR region. This has been performed for two cases of $N_r=16$ and $N_r=100$. Furthermore, the achievable common-rate corresponds to uniform power allocation (UPA) has been given. As can be seen, the theoretical results match well with the search method solutions in both cases. Moreover, as we expected, the proposed optimal power allocations outperform the UPA.



\section{Conclusion}

In this letter, we proposed a framework to obtain the optimal solution of various optimization problems including weighted sum-rate and common-rate optimization problems subject to a few conditions on the feasible SNR region. This framework does not restrict the region to be convex.

 To verify our analyses, we presented the optimal power allocations for a practical example of a two-way relay network under the assumption of large-scale antennas for the relay and sum power constraint. Finally, simulation results verified the optimality of the obtained theoretical solutions. One benefit of obtaining optimal solutions in closed form is that it enables further studies such as ergodic sum-rate, outage or error rate analysis. .



\begin{figure}[t!]
	\centering
	\includegraphics[width=21pc ]{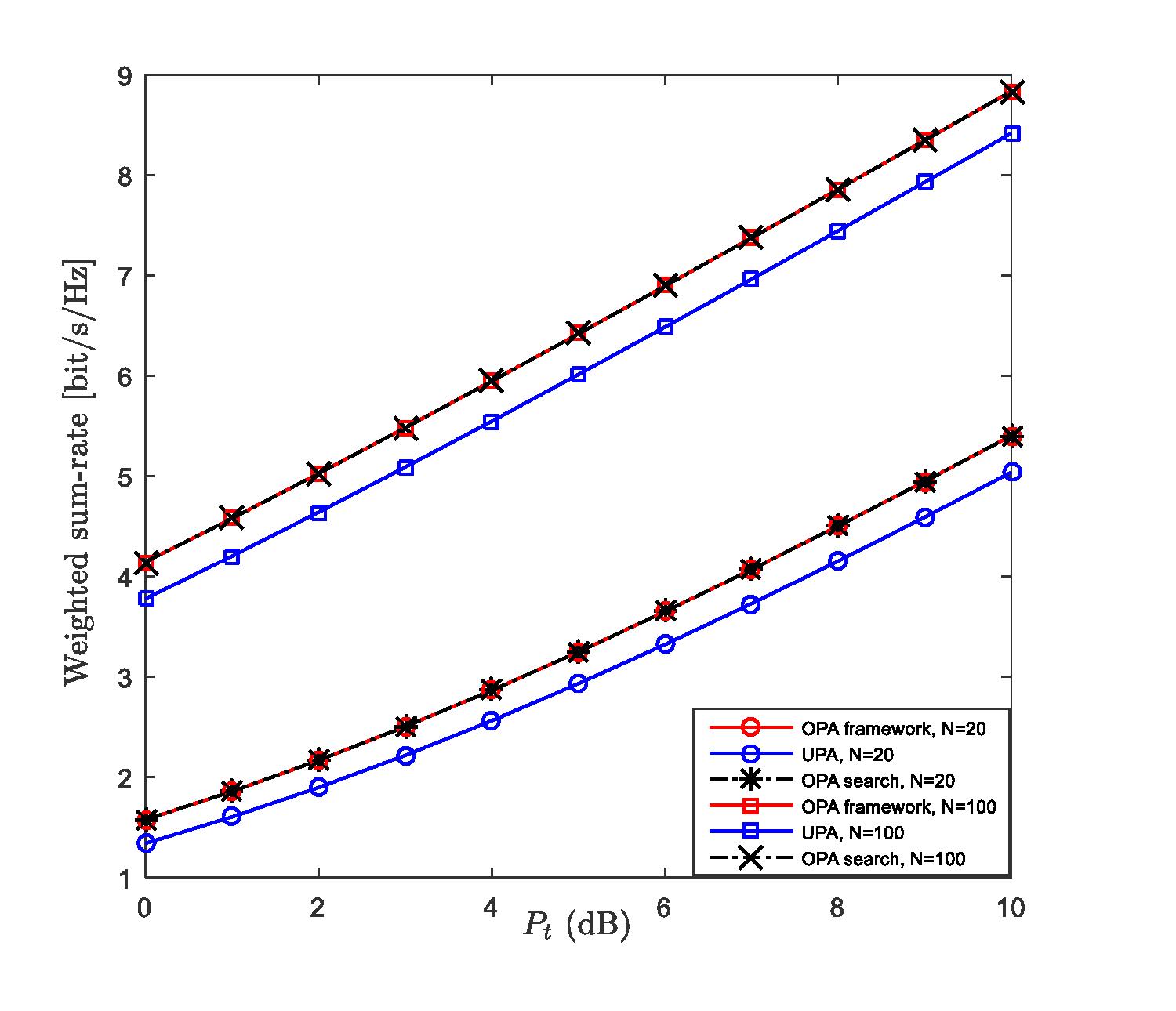}
		\vspace{-2.5em}
	\caption{Achievable weighted sum-rate  $a_1=2, a_2=1, \sigma_1^2=0.25, \sigma_2^2=1$ and $\sigma^2=1 $. }	
\end{figure}


%


\ifCLASSOPTIONcaptionsoff
  \newpage
\fi



%
\bibliographystyle{ieeetran}
\bibliography{Reference}{}

%




\end{document}